\documentclass[twocolumn]{autart}    

\usepackage{graphicx}          
\usepackage{amsmath}
\usepackage{algorithm}
\usepackage{algorithmic}
\usepackage{amssymb}
\usepackage{tikz}

\let\theoremstyle\relax
\usepackage{amsthm}

\newtheorem{assumption}{\textbf{Assumption}}
\newcommand{\RNum}[1]{\uppercase\expandafter{\romannumeral #1\relax}}

\newcommand{\CR}[1]{{\color{black}{#1}}}
\newcommand{\WZ}[1]{{\color{black}{#1}}}
\newcommand{\DI}[1]{{\color{black}{#1}}}

\let\classAND\AND
\let\AND\relax
\usepackage{algorithmic}

\let\AND\classAND
\usepackage{etoolbox}
\AtBeginEnvironment{algorithmic}{\let\AND\algoAND}

\newtheorem{lemma}{\textbf{Lemma}}
\newtheorem{remark}{\textbf{Remark}}
\newtheorem{theorem}{\textbf{Theorem}}

\begin{document}

\begin{frontmatter}

\title{Resilient Average Consensus: A Detection and Compensation Approach\thanksref{footnoteinfo}} 

\thanks[footnoteinfo]{This paper was presented in part at 60th IEEE Conference on Decision and Control, Austin, Texas, 2021 \cite{zhengaccurate}.}

\author[Shanghai]{Wenzhe Zheng}\ead{wzzheng@sjtu.edu.cn},    
\author[Shanghai]{Zhiyu He}\ead{hzy970920@sjtu.edu.cn},               
\author[Shanghai]{Jianping He}\ead{jphe@sjtu.edu.cn},  
\author[Hangzhou]{Chengcheng Zhao}\ead{chengchengzhao@zju.edu.cn},
\author[Shanghai]{Chongrong Fang}\ead{crfang@sjtu.edu.cn}
\address[Shanghai]{Department of Automation, Shanghai Jiao Tong University, and Key Laboratory of System Control and Information Processing, Ministry of Education of China}  
\address[Hangzhou]{State Key Laboratory of Industrial Control Technology, Zhejiang University}


\begin{abstract}
We study the problem of resilient average consensus for multi-agent systems with misbehaving nodes.
\WZ{To protect consensus value from being \CR{influenced} by misbehaving nodes, we address this problem by detecting misbehaviors, mitigating \CR{the corresponding adverse} impact and achieving \CR{the resilient} average consensus.}
\CR{In this paper, general} types of misbehaviors are considered, including deception attacks, accidental faults and link failures. 
We characterize \CR{the adverse impact} of misbehaving nodes in a distributed manner via two-hop communication information and 
develop a \underline{d}eterministic \underline{d}etection-\underline{c}ompensation-based \underline{c}onsensus (D-DCC) algorithm with a decaying fault-tolerant error bound. Considering scenarios where information sets are intermittently available \WZ{due to link failures}, a stochastic extension named \underline{s}tochastic \underline{d}etection-\underline{c}ompensation-based \underline{c}onsensus (S-DCC) algorithm is proposed.
We prove that D-DCC and S-DCC allow nodes to asymptotically achieve \WZ{resilient} average consensus exactly and in expectation, respectively. 
Then, the Wasserstein distance is introduced to analyze the accuracy of S-DCC\@.  
Finally, extensive simulations are conducted to verify the effectiveness of the proposed algorithms.
\end{abstract}

\end{frontmatter}




\section{INTRODUCTION}
Consensus problems have attracted extensive interests due to their wide applications, e.g., distributed control \cite{desai2001modeling}, estimation \cite{6172233} and optimization in robotic networks \cite{wang2014distributed}, smart grids \cite{8279503} and wireless sensor networks\cite{he2013sats}.
The goal of consensus is to enable nodes to reach global agreements via local exchanges of information by predefined consensus protocols. 
Under this framework, many interesting topics have been studied, including average consensus \cite{olfati2004consensus}, consensus with noises \cite{8440764} or packet drops \cite{hadjicostis2015robust}, and the convergence rates \cite{060678324}. 

Most of the aforementioned \WZ{studies} assume that all nodes faithfully execute predefined protocols.
Since distributed systems are usually deployed in open environments, there may exist vulnerabilities prone to failures or attacks\WZ{~\cite{dibaji2019systems}}, which could affect the consensus process. 
Hence, it is essential to constrain \WZ{the} negative impact \CR{brought by malicious or faulty nodes} to ensure desired agreements.
Motivated by this issue, numerous efforts have been devoted to resilient consensus \cite{6481629,Pasqualetti201290}, which are mainly divided into the following two categories.
The first category is based on Mean-Subsequence Reduced (MSR) algorithms \WZ{\cite{kieckhafer1993low,262588}}.
\WZ{The main idea of these studies is that the node discards the extreme states from neighbors and update its state with the remaining ones.}
An extended version named Weighted-Mean-Subsequence Reduced (W-MSR) \WZ{algorithm} is developed in~\cite{6315661}. \WZ{Instead of removing all extreme states as in the MSR algorithm, nodes only discard the extreme states that are strictly larger or smaller than their own states with W-MSR algorithm.} 
Based on the W-MSR algorithm, LeBlanc {\em et~al.}~\cite{6481629} develop a novel graph-theoretic property termed network robustness to characterize the resilience of W-MSR algorithms and derive the sufficient and necessary conditions of resilient asymptotic consensus. 
\WZ{In addition}, the quantized version of the W-MSR algorithm handling  asynchronous and time-varying time delays is presented in \cite{8100900}.
In faulty asynchronous networks, an approximate Byzantine consensus algorithm based on MSR algorithms is proposed in \cite{haseltalab2015approximate}. 
In applications, the MSR algorithm can be used to build resilient robot teams. \CR{For instance,} Saulnier {\em et~al.}~\cite{7822915} present a hybrid algorithm that enables resilient formation control for mobile robot teams in the presence of noncooperative robots. 
Guerrero {\em et~al.}~\cite{guerrero2017formations} propose algorithms to build robust robot teams to ensure the effectiveness of the MSR algorithm.
Similar to MSR algorithm, Yan {\em et~al.}~\cite{yan2020resilient} solve the resilient multi-dimensional consensus problem by a middle-point-based algorithm where these points are convex combinations at one dimension each time.  
The above algorithms ensure the consensus among normal nodes with no need to detect malicious nodes and do not require much computation cost. 
However, most of them put forward strict requirements on the graph structure, e.g., ($F+1,F+1$)-robust for $F$-total malicious model and ($2F+1$)-robust for $F$-local malicious model \cite{6481629,8100900}, \CR{which may be hard to satisfy in practice}.

The other category is detection and isolation.
The main idea  is  designing detection algorithms to find abnormal nodes and then isolate them immediately.
Observer-based techniques  are effective to achieve detections.
Pasqualetti {\em et~al.}~\cite{Pasqualetti201290} solve the fault detection and isolation problem based on observations for synchronous consensus in directed networks.
To solve problems with more general attack models, mobile agents  are exploited as observers  \cite{8013830}. This approach relaxes the requirements that the tolerable amount of attacks is strictly limited by the network connectivity and malicious nodes could not be neighbors.
Fault detection is also achieved by observer-based techniques for interconnected second-order systems \cite{SHAMES20112757}.
Detection method proposed in~\cite{sun2021security} is able to distinguish privacy noises and attacks with bounded false rate.
The method is to realize both privacy preserving and security under malicious attacks.
The detection and mitigation in randomized gossiping algorithms based on observation of temporal or spatial difference are investigated in systems of data injection attacks \cite{gentz2016data}.
Silvestre {\em et~al.}~\cite{silvestre2017stochastic} propose detection algorithms for randomized gossip algorithms by the use of Set-Valued Observers and Stochastic Set-Valued Observers. This method only requires a finite number of vertices to represent poly-topic sets and it can reduce the computational cost.
Observer-based techniques do not require additional state information of the system, but usually require a high computational cost.
Multiple communication information, e.g., two-hop information, also contribute to detection.
In~\cite{guo2012distributed}, the authors combine two communication protocols, i.e., communication-based and sensing-based model, for fault detection.
In~\cite{ramos2020general}, consensus algorithms are enabled in multiple distinct subsets of nodes, and each normal node verifies whether its state vector has different values in distinct subsets. If not, there will be malicious nodes among the subsets which have the same values.
Based on the majority voting, a detection scheme is designed with two-hop communication information for which the constraint on graph structures is less stringent than that for MSR algorithms \cite{yuan2021secure}.
Multiple communication information utilize extra state information to form redundancy relationships to determine whether the system is under attack.

The above researches present effective methods to achieve consensus among normal nodes. 
Due to the considerations of communication and computation costs, most of them cannot guarantee exact convergence to the average of initial states of nodes.
Specifically, MSR-based algorithms only ensure consensus within the range or convex hull of normal nodes' initial states, and detection-isolation-based algorithms do not eliminate \WZ{adverse impacts} caused by malicious nodes before isolation. 
Hence, the information of initial states may be polluted. Also, the latter category of algorithms might mistake faulty nodes for malicious ones when  link failures or miscalculations occur accidentally and will not adopt attack recovery methods, resulting in loss of information and system capacity. 
The previous work of resilient average consensus takes unreliable  heterogeneous communication links into consideration, but \CR{does not consider misbehaving nodes which include malicious and faulty ones}\cite{6426666}.

\CR{Therefore, in this paper, we are motivated to design resilient average consensus algorithms for multi-agent systems to defend against the adverse impact on the consensus that are brought by misbehaving nodes. 
To achieve better performance of resilient average consensus, we adopt the idea of detecting and compensating the adverse impact of misbehaving nodes with two-hop communication information and providing tolerance for faulty nodes.
We designed detection methods, estimate and compensation the errors in a distributed manner even if the two-hop information is unreliable when malicious nodes are not neighbors in our previous work \cite{zhengaccurate}. In this paper, we make further investigation and the main contributions of this paper are summarized as follows:
}

\begin{itemize}
\item We study the problem of resilient average consensus with misbehaving nodes. Not only detection and isolation, but also estimation and compensation are adopted to achieve accurate \WZ{resilient} average consensus.
\item By utilizing two-hop communication information, we design a \underline{d}eterministic \underline{d}etection-\underline{c}ompensation-based \underline{c}onsensus (D-DCC) algorithm for normal nodes to detect and compensate the adverse impact from misbehaving nodes. We also prove that \WZ{resilient} average consensus can be achieved by D-DCC exactly.
\item We further consider the scenario where the communication links could fail due to accidents, which brings extra challenges for the accurate resilient average consensus problem. In this case, we propose a \underline{s}tochastic \underline{d}etection-\underline{c}ompensation-based \underline{c}onsensus (S-DCC) algorithm correspondingly. Compared with our previous work\cite{zhengaccurate}, in this paper, we relax the requirement that the expectation of faults should be zero in S-DCC, and prove that \WZ{resilient} average consensus in expectation is achieved by S-DCC.
Moreover, we analyze the accuracy of S-DCC by Wasserstein distance, and present extensive evaluation results to show the effectiveness of the proposed algorithms.
\end{itemize}


The rest of this paper is organized as follows. 
Section 2 introduces the models of networks and \WZ{misbehaviors}. 
Section  3 presents the proposed algorithms consisting of the mechanisms of detection and compensation.
Then, the performance of the algorithms is analyzed in Section 4. 
In Section 5, simulation results are provided. 
Finally, Section 6 concludes the paper and discusses the future directions.

\WZ{\emph{Notation:}}
\WZ{We denote by $\mathbb{R}$ the set of reals.
Given a matrix $M$, $M^\top$ is its transpose.
Given a set $\mathcal{V}$, $|\mathcal{V}|$ is the number of elements in the set and $\mathcal{V} / \{ i \}$ is the set removing $i$ from $\mathcal{V}$.
Given a matrix $W$, $W_{ \{ i\}}$ is the matrix deleting $i$-th row and column of $W$.
We let $\mathbb{E}$ and $\mathbb{D}$ denote expectation and variance operations, respectively.

}

\section{PRELIMINARIES AND PROBLEM FORMULATION}
\subsection{Network Model}
Consider a network modeled as an undirected graph $\mathcal{G}=(\mathcal{V},\mathcal{E})$ with vertex set $\mathcal{V}=\{1,2,\WZ{\ldots} ,N\}$ and edge set ${\mathcal{E} \subseteq \mathcal{V} \times \mathcal{V}}$. Note that $(i,j) \in \mathcal{E} $ indicates \WZ{that} node $i$ and node $j$ can communicate with each other. The neighbor set of node $i$ is denoted by $\mathcal{N}_i=\{ j|(i,j) \in \mathcal{E}\}$. 
The adjacency matrix is $A_\mathcal{G}=[a_{ij}]_{N\times N}$, and Laplacian matrix is $L=D_\mathcal{G}-A_\mathcal{G} $, where $ D_\mathcal{G} = \mathrm{diag} (d_1,\ldots,d_N)$ with $d_i = \sum_{j = 1}^{N}a_{ij}$.
Let $d_m=\max\{d_1,\dots,d_N\}$.
\WZ{Without loss of generality, let the subsets} $ \mathcal{V}_s=\{1,2, \cdots ,n \}$ and $ \mathcal{V}_m=\{ n+1, n+2,\cdots, N \}$ represent normal nodes set and misbehaving nodes set, respectively. 
It follows that $\mathcal{V}_s \cup \mathcal{V}_m=\mathcal{V}$ and $\mathcal{V}_s \cap \mathcal{V}_m=\emptyset$.

\subsection{Consensus Algorithms}
Let $x_i(k) \in \mathbb{R}$ be the state of node $i$ at time $k$ and ${x}(k)=[x_1(k), x_2(k),\WZ{\ldots},x_N(k)]^\top$ be the \WZ{state vector}.
Consensus algorithms are distributed control protocols that drive all states to the same value, \WZ{i.e.,} $x_i=c,$ $\forall i \in \mathcal{V}$ \cite{olfati2004consensus}. 
The value $c$ is called \WZ{the} consensus value. Particularly, if $c= \WZ{\frac{1}{N}} \sum_{i=1}^N x_i(0)$ \WZ{holds}, \WZ{then the} average consensus \WZ{can be} achieved.
 The basic discrete-time linear average consensus is represented as
\begin{equation}\label{e1}
x(k+1)=Wx(k),
\end{equation}
where $W$ is a doubly stochastic matrix, and $w_{ij}\neq 0$ if and only if nodes $i$ and $j$ are neighbors. 
By (\ref {e1}), the system will achieve \WZ{the} average consensus exponentially if \WZ{the undirected graph} $\mathcal{G}$ is a connected graph. 
The commonly used weight matrix guaranteeing asymptotic convergence includes Metropolis weights\cite{xiao2005scheme} and Perron matrix, i.e., $W=I-\gamma L$, where $0<\gamma<1/d_m$.
\CR{Under \WZ{both Metropolis weights} and Perron weights, the updating coefficients of node $i$, i.e., $\WZ{w}_{ij}$, $\forall j \in \mathcal{N}_i$, are known by neighbors if the number of neighbors $|\mathcal{N}_i|$ and $N$ are \WZ{available to} neighbors, which provide bases for the later error detection.}

\subsection{Information Set and \WZ{Misbehavior} \CR{Model}}\label{subsec info set}

The rule (\ref{e1}) implies that each node updates its state by using the states of its own and its neighbors. 
Such two-hop information can be properly utilized to facilitate the efficient detection of misbehaviors of targeted nodes \cite{he2013sats,8013830}.
All nodes transmit \WZ{their own} information sets to neighbors at each time.
Specifically, for node $i$ at time $k$, \WZ{its} information set $\Psi_i(k)$ is denoted by 
\begin{equation*}
\Psi_i(k) \!= \! \{i, x_i(k),\pi_i(k) , \varepsilon_i(k-1), \{j, x_j^{(i)}(k-1) , j \in \mathcal{N}_i\}\},
\end{equation*}
where $x_j^{(i)}(k-1)$ is the state value (at time $k-1$) of  node $j$ \WZ{which is}  sent by node $i$ to \WZ{its} neighbors at time $k$. 
\WZ{The term} $\varepsilon_i(k)$ is the compensation added by normal nodes \CR{or the adverse impact brought by misbehaving nodes, which} will be discussed \CR{in detail later}. 
\CR{Let the \WZ{binary} attack detection \WZ{indicator} $\pi_i(k)=1$ or $\pi_i(k)=0$ represent ``attack" or ``no attack", respectively, where the former indicates that the node $i$ \CR{has} detected misbehaving nodes in its neighbors.} 
\CR{Note that $\varepsilon_i(k\WZ{-1})$ is allowed to be non-zero for normal nodes if $\pi_i(k) = 1$, i.e., compensation is added only when a misbehavior is detected.}

\CR{As for misbehavior model,} the misbehaving nodes considered in this paper can be either malicious or faulty.
Faulty nodes may cause \CR{adverse impacts} to the system because of accident faults, e.g., miscalculations.
\WZ{A malicious node aims to disrupt the network functions by manipulating the information set, but can only send the same information to all of its neighbors at each time.}
\CR{If the network is realized by broadcast communication, it is natural to assume that any node sends the same value to all of its neighbors.}
\CR{In the following,} we make three assumptions regarding the specific \WZ{misbehaviors} that \WZ{misbehaving} nodes can generate. 
\begin{assumption}\label{ass1}
Any two misbehaving nodes do not neighbor with each other.
\end{assumption}
\begin{assumption}\label{ass2}
A malicious node can manipulate its information set by changing \CR{the state values of its own and  its neighbors, and delete the IDs and states of its neighbors,} but cannot add any entries.
\end{assumption}
\begin{assumption}\label{ass3}
A normal node will cut off all future communication with \CR{the node(s) that is(are) detected as abnormal and isolated.} 
\end{assumption}
Assumption \ref{ass1} is reasonable when the number of misbehaving nodes is much less than that of normal nodes. \CR{In another word, the} misbehaving nodes are sparsely distributed in the network~\cite{he2013sats}. 
The \CR{attacker capabilities} considered are specified in Assumption \ref{ass2},
\WZ{which includes basic deception attacks such as spoofing attack and false-data injection attack\cite{dibaji2019systems}.}
Malicious nodes prefer not to add any entries because it will be easily detected by \CR{normal nodes with two-hop information.} 
With Assumption \ref{ass3}, the misbehaving nodes will be effectively isolated by all normal nodes, which can be achieved when there are mobile nodes in the network \WZ{\cite{8013830}}.



Since misbehaving nodes will cause \CR{adverse impacts} and normal nodes will add compensation input, 
the discrete-time linear updating rule is given as follows
\begin{equation}\label{e2}
x(k+1)=Wx(k)+\varepsilon(k),
\end{equation} 
where $\varepsilon(k) = \WZ{[\varepsilon_1(k), \varepsilon_2(k),\ldots,\varepsilon_N(k)]^\top}$ is the input vector. 
\WZ{The term $\varepsilon_i(k)$ is the error \CR{input} for misbehaving \CR{node} $i\in \mathcal{V}_m$, while  $\varepsilon_i(k)$ is the  compensation input for normal \CR{node} $i \in \mathcal{V}_s$.}

\WZ{Considering there may be link failures during communication, we refer to link failures as the phenomenon that may prevent information set from being received at each desired time. We denote by $p$ the probability of connection between nodes, i.e., link failure occurs with probability $1-p$ between each pair of  nodes independently.}
\WZ{In order to facilitate the analysis when we consider link failures,  the errors caused by misbehaving nodes are characterized to obey an unknown distribution~\cite{4663901}.
It is supposed that misbehaving node $i$ affects the system (or the error equals to zero) with probability $\theta_i \in [0,1]$.
Let $X_i(k)$ be the random variable of attack or not, i.e., $X_i(k) \sim \mathcal{B}(1,\theta_i)$, where $\mathcal{B}$ represents the Bernoulli distribution.
The term $X_i(k)$ is used because we need to consider the probability that $\varepsilon_i(k)=0$ separately.
Further, the misbehaving nodes affect the system with a certain mechanism.
Hence, if node $i$ affects the system at time $k$, the random variable of the error, i.e., $Y_i(k)$, obeys a certain distribution with expectation $\mu_i$ and variance $\sigma_i^2$. 
We assume that the probability of $Y_i(k) = 0$ is zero, and $X_i(k)$ and $Y_i(k)$ are independent.
Hence, it holds that $\varepsilon_i(k) = X_i(k)Y_i(k)$ for misbehaving nodes.
Then, we have
\begin{equation*}
\begin{aligned}
\mathbb{E} [ \varepsilon_i(k) ]  ~&= \theta_i\mu_i,\\ 
\mathbb{D} [ \varepsilon_i(k) ]   ~&= \theta_i\sigma_i^2+(1-\theta_i)\theta_i\mu_i^2 \triangleq \sigma_{\varepsilon_i}^2.
\end{aligned}
\end{equation*}
Nevertheless, $\varepsilon_i(k)$ can obey an arbitrary distribution because we do not pose any restriction on the distributions of $X_i(k)$ and $Y_i(k)$, and do not need to know the expectation and variance of them, which is different from the faults.
The difference between malicious nodes and faulty nodes is that malicious nodes will attack the system continuously, while faulty nodes only cause \CR{accidental disturbances} in a limited period.
}

\subsection{Problem of Interest}
We consider a multi-agent system with misbehaving nodes which is described by $\mathcal{G}$. Each node owns an initial state $x_i(0)$ and the system updates its states by (\ref{e2}). 
\WZ{In such a setting}, representative consensus algorithms \WZ{such as W-MSR \cite{6315661}} cannot guarantee accurate consensus values, and fault detection and isolation methods may mistake faulty nodes for malicious ones.
Thus, extra input is needed to mitigate errors \WZ{caused by misbehaving nodes}, and a fault-tolerance mechanism is \WZ{called} for faulty nodes. 
This paper aims to develop a distributed detection and compensation algorithm to achieve resilient average consensus.
\WZ{With the proposed algorithm}, normal nodes \CR{can} detect the errors of neighboring misbehaving nodes by examining the information sets from neighbors and mitigate the \CR{adverse} impact by adding compensating input to their own states.
\WZ{Considering that there may be misbehaving nodes with low data utility, e.g., malicious nodes who constantly cause errors and faulty nodes with severe malfunction, isolation are adopted to thoroughly eliminate negative effects.} 
\CR{Besides, given that the communication in multi-agent systems could be unreliable due to link failures (especially in the wireless communication scenario), the consensus process will be influenced if the information set is intermittently unavailable. Therefore, \WZ{we further take link failures into consideration.}}

Under scenarios where all information sets from neighbors are available, we aim to design a \CR{misbehavior-resilient algorithm  to achieve \WZ{resilient} average consensus \WZ{among the nodes in the set after isolation $\mathcal{V}_r$} for the system with malicious nodes and faulty nodes, i.e.,
\begin{equation}\label{ee}
\lim_{k\to \infty} x_j(k)= \frac{1}{|\mathcal{V}_r|} \sum_{u\in \mathcal{V}_r} x_u(0),\quad \forall j\in \mathcal{V}_r,
\end{equation} 
where $\mathcal{V}_r$ is a subset of $\mathcal{V}$ including nodes that are not isolated}.
It is assumed that the subgraph ${\mathcal{G}_r}=( \mathcal{V}_r , \mathcal{E}_r)$  is connected, where $\mathcal{E}_r$ \WZ{($\mathcal{E}_r \subseteq \mathcal{E}$)} denotes \WZ{edge set of the} nodes in $\mathcal{V}_r$.

\WZ{Considering the random communication link failures between nodes, we aim to extend our algorithm to achieve unbiased \WZ{resilient} average consensus in expectation with a relative small variance by designing compensation input, i.e., }
\begin{equation} \label{th2.1}
\mathbb{E} \left[  \lim_{l\to \infty} x_j(l) \right] = \frac{1}{|\mathcal{V}_r|} \sum_{u\in \mathcal{V}_r} x_u(0),\quad \forall j\in \mathcal{V}_r.
\end{equation}

\section{DETECTION AND COMPENSATION DESIGN}
In this section, we \WZ{first propose a detection algorithm to detect misbehaving nodes and then compensate the negative impact caused by these misbehaviors.}
The basic idea \WZ{of detection} is to design two-hop information set \WZ{to form redundancy
relationships}.
By two-hop information set, \WZ{the misbehaviors} can be characterized  in a distributed manner, which leads to four compensation schemes.
Considering the possible link failures, a stochastic scheme is then introduced. 


\subsection{Detection Strategies} 
The first step for each normal node is to determine whether there are misbehaving nodes in the neighborhood and \CR{estimate the amount of error} caused by them. 
According to Assumption \ref{ass2}, a malicious node $i$ can manipulate states of neighbors in the information set, i.e., $x_j^{(i)}(k-1)\WZ{,j\in \mathcal{N}_i}$, or updates its own state with arbitrary \WZ{errors}. 
The detection strategies are characterized by the following two types:
\begin{itemize}
\item  \emph{Detection Strategy \RNum{1}}: \WZ{N}ode $j$ detects whether misbehaving nodes change the state values of $j$ in information set, i.e., $x^{(i)}_j(k) \ne x_j(k), j \in \mathcal{N}_i$. 
If the malicious node deletes the ID and state of $j$ in the information set, it can be regarded as changing the corresponding state value to zero.
\item  \emph{Detection Strategy \RNum{2}}: \WZ{Node $j$ detects whether misbehaving node $i$ follow the update rule $x_i(k+1) = \sum_{j\in \mathcal{N}_i} w_{ij}x^{(i)}_j(k)$.}
\end{itemize}

Each normal node will utilize \WZ{the} two detection strategies to \WZ{check} each neighbor \CR{based on its} information set.
\CR{Obviously, the misbehaviors mentioned in Section \ref{subsec info set}} could be detected by Detection Strategy \RNum{1} \WZ{or} \RNum{2} or both. The update rule is as
\begin{equation}\label{e4-1}
\begin{split}
x_i(k+1) ~&=\sum_{j \in \mathcal{N}_i} w_{ij}x_j(k)+\varepsilon_i(k)\\
 ~& =\sum_{j \in \mathcal{N}_i} w_{ij} x_j(k)+\sum_{j\in \mathcal{N}_i} \varepsilon_i^{j(1)}(k) +\varepsilon_i^{(2)}(k),
\end{split}
\end{equation}
where $\varepsilon_i^{j(1)}(k) $  and $\varepsilon_i^{(2)}(k)$ are the \CR{adverse impacts} detected by Detection Strategy \RNum{1} and \RNum{2}, respectively, and
 \begin{subequations}\label{enoise}
\begin{align}
 \varepsilon_i^{j(1)}(k) ~& = w_{ij} (x_j^{(i)}(k)-x_j(k)) ,\label{e4-b}\\
\varepsilon_i^{(2)}(k) ~& = x_i(k+1)-\sum_{j\in \mathcal{N}_i} w_{ij}x_j^{(i)}(k).\label{e4-c}
\end{align}
\end{subequations}
\DI{Note that each neighbor of misbehaving node $i$ will detect node $i$ with the same error by Detection Strategy \RNum{2}. Hence, $\varepsilon_i^{(2)}(k)$ is used \WZ{instead of} $\varepsilon_i^{j(2)}(k)$.}

\subsection{Compensation Schemes \& D-DCC Algorithm} \label{D-DCC}
In this part, a \underline{d}eterministic \underline{d}etection-\underline{c}ompensation-based \underline{c}onsensus (D-DCC) algorithm is proposed.
The following lemma shows the sufficient condition of resilient average consensus on dynamic system (\ref{e2}).
\begin{lemma}[\hspace{1sp}\cite{8356738}] \label{lemma:1}
For the system (\ref{e2}), if the added input vectors are bounded, i.e., $||\varepsilon(k)||_\infty \le \alpha \rho^k$ for certain $\alpha>0$ and $\rho \in [0,1)$, and the sum of inputs satisfies 
\begin{equation}\label{th1}
\sum_{k=0}^{\infty} \sum_{i=1}^\WZ{N} \varepsilon_i(k)=0,
\end{equation}
then average consensus is achieved exponentially.
\end{lemma}

It is obvious that the existence of misbehaving nodes can lead to the violation of (\ref{th1}), which is the necessary condition of average consensus. If (\ref{th1}) does not hold, $\lim_{k \to \infty}\sum_{i=1}^{N} x_i(k)= \sum_{i=1}^{N} x_i(0) $ will not hold either.
To achieve exact average consensus, we need to compensate the impact of misbehaviors by introducing an error compensator $\eta_j$ for each normal node $j$.
The compensation values to be added is stored in the error compensator $\eta_j$.
\WZ{Then, node $j$ will select compensation input according to the error compensator.}
We define the following three types of compensation.
\begin{itemize}
\item  \emph{Compensation Scheme \RNum{1}}: To compensate the impact detected by Detection Strategy \RNum{1}, i.e.,
\begin{equation}\label{com1}
\eta_{j}^{i(1)}(k+1)= -w_{ij} (x_j^{(i)}(k)-x_j(k)).
\end{equation}
\item  \emph{Compensation Scheme \RNum{2}}: To compensate the impact detected by Detection Strategy \RNum{2}, i.e., 
\begin{equation}\label{com2}
\eta_{j}^{i(2)}(k+1)= -\varepsilon_i^{(2)}(k)/|\mathcal{N}_i|.
\end{equation}
\item  \emph{Compensation Scheme \RNum{3}}: To compensate the impact of isolation, i.e., 
\begin{equation}\label{com3}
\eta_{j}^{i(3)}(k+1)=(x_i(k+1)-x_i(0))/|\mathcal{N}_i| .
\end{equation}
\end{itemize}
Compensation \RNum{1} is employed by node $j$ when it detects that its neighbor $i$ has changed the value of node $j$ in the information set.
Note that all $|\mathcal{N}_i|$ neighbors of node $i$ detect the misbehaviors \WZ{of node $i$} by Detection Strategy \RNum{2}.
Hence, each node averagely compensates the error.
\WZ{Moreover, Compensation Scheme \RNum{3} is adopted when node $i$ is isolated by neighbors.}
Each neighbor will equally compensate the historical \CR{adverse} impact on average consensus.

Inspired by Lemma \ref{lemma:1}, we adopt a distributed exponential decaying bound of errors, i.e., $\alpha_j\rho_j^k$, $j\in \mathcal{V}_s$, to decide isolation and guarantee the convergence. 
Let $\alpha= N\max_{i\in \mathcal{V}}   \alpha_i  $, $\rho = \max_{i \in \mathcal{V}} \rho_i$. Then, the condition
$||\varepsilon(k)||_\infty \le \alpha \rho^k$ will hold.
Node $j$ \WZ{estimates the error of its neighbor $i$ by (\ref{enoise})}. 
If the error is in the bound, \WZ{i.e. $|\varepsilon_i^{j(1)}(k) + \varepsilon_i^{(2)}(k)| \le \alpha_j\rho_j^k$}, then node $j$ will compensate the error by Compensation Scheme \RNum{1} and \RNum{2}, which is a resilient mechanism for finite errors and accidental errors such as computation error and actuator error.
Otherwise, \CR{it means that the adverse impact caused by node $i$ is too severe such that the consensus process will be seriously affected and the convergence may not be achieved according to Lemma \ref{lemma:1}. In this case, node $j$ will cut off the communication with node $i$ (i.e., node $i$ is isolated) to avoid the future adverse impact from node $i$ and Compensation Scheme \RNum{3} will be used to remedy the historical bad impact from node $i$.}

\CR{In summary, by means of the above three compensation schemes, our resilient average consensus algorithm under deterministic communication scenario is proposed as D-DCC algorithm, which is summarized in Algorithm \ref{algo-1}. For the sake of simplicity, the proposed D-DCC algorithm only shows the execution of node $j\in\mathcal{V}$ and its neighbors.}
\WZ{Specifically,} node $j\in\mathcal{V}$ first \CR{performs Detection Strategies \RNum{1} and \RNum{2}, and Compensation Schemes \RNum{1} and \RNum{2} in order (see steps \ref{line:l1}-\ref{line:l4} in Algorithm \ref{algo-1}).} 
Then, node $j$ checks whether the error of node $i\in\mathcal{N}_j$ beyonds the bound $\alpha_i \rho_i^k$.
Subsequently, node $j$ identifies \WZ{whether node $i$ is to be isolated at this time}. If so, \WZ{node $j$ calculates} $\eta_{j}^{i(3)}(k+1)$ by Compensation Scheme \RNum{3} and updates the new $|\mathcal{N}_j|$ to its remaining neighbors.
In the sequel, node $j$ selects \WZ{$\varepsilon_j(k+1)$} and renewals its state $x_j(k+2)$ with the designed $\varepsilon_j(k+1)$, followed by updating the error compensator $\eta_j=\eta_j-\varepsilon_j(k+1)$. \CR{Note that the selection of the compensation input $\varepsilon_j(k+1)$ could be arbitrary, as long as it guarantees the \WZ{security of} state and the non-increasing property of $|\eta_j|$.}

\begin{remark}
\CR{It should be noted that both normal and misbehaving nodes have the input term $\varepsilon_i(k)$ (i.e., the error or compensation input) in their information sets. If malicious nodes have full knowledge of detection and compensation methods, then they can easily masquerade as normal ones. To avoid this issue, we put the attack detection \WZ{indicator} $\pi_i(k)$ into the information set $\Psi_i(k)$. Note that a normal node is allowed to add non-zero compensation input only when it has detected misbehaviors in the neighborhood.} 
\DI{Then, we adopt a steady compensation sequence which restricts the changes of compensation, i.e., $|\varepsilon_i(k)-\varepsilon_i(k-1)| \le \delta$, \WZ{which is reasonable in practice.} We assume that malicious nodes cannot change the attack detection \WZ{indicator} or have no knowledge of $\delta$. \WZ{Hence, with the above two methods, it may be easy to distinguish malicious nodes with errors and normal nodes with compensation input.}}
\end{remark}
\begin{algorithm}[t]
 \small
    \caption{D-DCC Algorithm}
    \label{algo-1}
    \begin{algorithmic}[1]
    \REQUIRE~~{$\ \ x_j(k), \Psi_i(k+1), i \in \mathcal{N}_j$}.
    \STATE \textbf{Initialize}: 	{set compensator $\eta_j=0$}, {parameters $\alpha_j,\rho_j$}\;
node $j$ exchanges its true initial state $x_j(0)$ and the number of neighbors $|\mathcal{N}_j| $ with neighbors, and store them\;
\FOR {$k=0 : $ Max\_time}
    \FOR {$i \in \mathcal{N}_j$ \WZ{(node $i$ is not isolated \& $\pi_i(k)=0$)}}
	 \STATE{$\varepsilon_i^{j(1)}(k+1) =w_{ij}(x_j^{(i)}(k)-x_j(k))$}.   \label{line:l1}
	 \STATE{$\varepsilon_i^{(2)}(k) = x_i(k+1)-\sum_{j\in N_i} w_{ij}x_j^{(i)}(k) $}.  \label{line:l3}
	\IF  {$\varepsilon_j^{i(1)}(k) \neq 0$}
     \STATE $\eta_{j}^{i(1)}(k+1)= -\varepsilon_i^{j(1)}(k)$. 
      \ENDIF \label{line:l2}
	\IF  {$\varepsilon_i^{(2)}(k)  \neq 0$}
    { \STATE $\eta_{j}^{i(2)}(k+1)= -\varepsilon_i^{(2)}(k)/|\mathcal{N}_i|$.
    }
      \ENDIF \label{line:l4}
		\IF {$|\varepsilon_i^{j(1)}(k)+\varepsilon_i^{(2)}(k)| >\alpha_j \rho_j^k $} \label{line:l5}
    { \STATE Node $j$ cut off the communication with node $i$. 
    }
      \ENDIF \label{line:l6}
			\IF {node $i$ is isolated at this time}\label{line:l7}
    { \STATE  $\eta_{j}^{i(3)}(k+1)=(x_i(k+1)-x_i(0))/|\mathcal{N}_i|$. 
    }
      \ENDIF \label{line:l8}
	\STATE \CR{$\eta_j=\eta_j+\eta_{j}^{i(1)}(k+1)+\eta_{j}^{i(2)}(k+1)+\eta_{j}^{i(3)}(k+1)$.} \label{line:l9}
    \ENDFOR


	\STATE Update its state $x_j(k+2)$ with $\varepsilon_j(k+1)$ by (\ref{e4-1}).
    \STATE \CR{Update the error compensator} $\eta_j=\eta_j-\varepsilon_j(k+1)$.
\ENDFOR
    \end{algorithmic}
\end{algorithm}

\subsection{S-DCC Algorithm}
\CR{
In this subsection, we further investigate the resilient average consensus problem against misbehaviors while considering the possible link failures among nodes. In this case, there are two aspects introducing stochasticity that should be addressed: (1) The communication link status between nodes is stochastic and we assume that the link failure occurs with probability $p$ at each time (see Section \ref{subsec info set}); (2) The compensation is also stochastic because of random misbehaviors and link failures. This stochasticity property brings extra challenges compared with the problems in the deterministic communication scenario (see Section \WZ{\ref{D-DCC}}). That is the corresponding information set of neighbors is not available when a link failure occurs. As a result, some misbehaviors among the relevant nodes may not be detected. To ensure the resilience performance against misbehaving nodes when the information set is randomly unavailable, we propose a \underline{s}tochastic \underline{d}etection-\underline{c}ompensation-based \underline{c}onsensus (S-DCC) algorithm to reach the average consensus in expectation with bounded variance. Specifically, we \WZ{further} propose Compensation Scheme \RNum{4} based on the estimation of the average \WZ{detected errors}:
}


\begin{itemize}
\item  \emph{Compensation Scheme \RNum{4}}: To compensate the impact of undetected misbehaviors, i.e.,
\begin{equation} \label{compensation4}
\eta_{j}^{i(4)}(k+1) = -(k-k_{i}^{j0}-m_j(k))\overline{\varepsilon}_{i}^{j}(k),
\end{equation}
\end{itemize}
where \WZ{$k_{i}^{j0}$ is the last detecting time before node $i$ is first detected by node $j$ as a misbehaving node}, which is treated as the last time before misbehavior occurs.
\WZ{In addition, we have} 
\begin{equation*}
\overline{\varepsilon}_{i}^{j}(k)= \sum_{\varepsilon_{i}^j(k) \in  \Omega_j^{(i)}(k)} \varepsilon_{i}^j(k)/m_j(k),
\end{equation*}
where $m_j(k)$ is the number of times that \WZ{node} $j$ detects node $i$ after time $k_{i}^{j0}$.
The intuition behind Compensation Scheme \RNum{4} is that the mean of detected errors may represent the mean effect by \WZ{misbehaving} nodes in a time window.
Specifically,  \CR{the information set is available, which enables the detection,} at each time independently with probability $p$. When the detection is enabled, node $j$ will detect and compensate  the possible errors. In order to estimate the effect of \WZ{misbehaving} nodes, node $j$ will store the detected error $\varepsilon_i^{j}(k)$ in the error set $\Omega_j^{(i)}(k)$ corresponding to node $i$, and compensate undetected errors according to the error set. 
\CR{The details of S-DCC algorithm are summarized in Algorithm \ref{algo-2}}.
\WZ{Note that at each time of detection, a new detected error will be added to the error set, and a new estimation of undetected error will be utilized to replace the former one. Hence, at step \ref{line:l10} of S-DCC, the former compensation by Compensation Scheme \RNum{4} will be removed.}

\begin{algorithm}[t]
 \small
    \caption{S-DCC Algorithm}
    \label{algo-2}
    \begin{algorithmic}[1]
    \REQUIRE~~{$\ \ x_j(k),\Psi_i(k+1),i \in \mathcal{N}_j $}.
    \STATE \textbf{Initialize}: The same as Algorithm 1;
       set detection probability $p$, detection times $m_j=0$.
\FOR {$k=0 : $ Max\_time}
{  \IF {Detection is enabled}
    \FOR {$i \in \mathcal{N}_j$ \WZ{(node $i$ is not isolated \& $\pi_i(k)=0$)}}
	\STATE Execute steps \ref{line:l1}-\ref{line:l6} in Algorithm \ref{algo-1}.
	\IF {node $j$ is misbehaving}
		\STATE {Store $\varepsilon_i^{j}(k)=\varepsilon_i^{j(1)}(k)+\varepsilon_i^{(2)}(k)/|\mathcal{N}_i|$} in $\Omega_j^{(i)}(k)$.
			\STATE {Calculate $\eta_{j}^{i(4)}(k+1) $ by (\ref{compensation4})}.
     \ENDIF
\IF {$i$ is isolated at this step}
    { \STATE  $\eta_{j}^{i(3)}(k+1)= (x_j(k+1)-x_j(0))/|\mathcal{N}_i|$. 
    }
      \ENDIF
	\STATE $\eta_j=\eta_j+\eta_{j}^{i(1)}(k+1)+\eta_{j}^{i(2)}(k+1)+\eta_{j}^{i(3)}(k+1)+\eta_{j}^{i(4)}(k+1)-\eta_{j}^{i(4)}(k)$. \label{line:l10}
    \ENDFOR
\ENDIF
}
	\STATE Update its state $x_j(k+2)$ with $\varepsilon_j(k+1)$ by (\ref{e4-1}).
	\STATE \WZ{Update the error compensator} $\eta_j=\eta_j-\varepsilon_j(k+1)$.
\ENDFOR
    \end{algorithmic}
\end{algorithm}

\section{Performance Analysis}
In this section, we prove that for D-DCC, all \WZ{misbehaving} nodes will be detected and accurate \WZ{resilient} average consensus will be achieved. Additionally, for S-DCC, we demonstrate that all misbehaving nodes will be detected with probability one and \WZ{resilient} average consensus in expectation will be achieved.
\WZ{Besides,} we analyze the accuracy of S-DCC by the Wasserstein distance. 

\subsection{Analysis of D-DCC}
First, \WZ{with respect to the detection performance of D-DCC, the} following lemma shows the effectiveness.
\begin{lemma}
If Assumptions \ref{ass1}-\ref{ass3} hold, then all misbehaving nodes will be detected and some of them will be isolated  by Algorithm 1. 
\end{lemma}

\begin{proof}
Suppose that the misbehaving node $i$ changes the state value of its neighbor node $j$ in the information set $\Psi_i(k+1)$. According to Assumption 1, node $j$ is normal. Hence, when node $j$ receives the information \WZ{set} from node $i$, it will find out that $x_j^{(i)}(k)\neq x_j(k)$ and \WZ{the misbehaving} node $i$ will be detected by Detection Strategy \RNum{1}. Similarly, if node $i$ deletes the ID and state of \WZ{node} $j$, it will \WZ{also} be detected by node $j$.

If the misbehaving node \WZ{$i$} does not follow the update rule based on the information set, i.e., $x_i(k+1)\neq \sum_{j\in \mathcal{N}_i} w_{ij}x^{(i)}_j(k)$, it will be detected by all neighbors by Detection Strategy \RNum{2}, because $w_{ij}$ is known by all neighbors.

According to Algorithm \ref{algo-1}, if the error is below the local bound, \WZ{then} the misbehaving node will not be isolated. However, once the error is out of the bound, the misbehaving node will be isolated. 
\end{proof}

Next, we evaluate the consensus performance of D-DCC algorithm in the presence of misbehaving nodes in the following theorem.

\begin{theorem}\label{th1.1}
If Assumptions \ref{ass1}-\ref{ass3} hold, then D-DCC achieves \WZ{resilient} average consensus \CR{among the nodes in the set after isolation $\mathcal{V}_r$}, i.e., (\ref{ee}) holds.
\end{theorem}
\begin{proof}
First, we illustrate that consensus will be achieved by Algorithm \ref{algo-1}.
Note that $W$ is doubly stochastic and \CR{each $\varepsilon_i(k),i\in\mathcal{V}_r$ satisfies the condition $\varepsilon_i(k) \le \alpha_i \rho_i^k$. Therefore, the consensus can be achieved according to Lemma 1 after compensating the impacts of nodes that are not in the set $\mathcal{V}_r$.}
If node $i$ is isolated, let $W_{ \{ i\}} $ be  the matrix obtained by  deleting $i$-th row and column of $W$. It follows that $W_{ \{ i\}}$ is still a matrix with Perron weight or Metropolis weights. Hence, consensus will be achieved among \CR{the remaining nodes in $\mathcal{V}_r$.}

Next, we prove the limit value is average consensus \CR{among the nodes in the set after isolation $\mathcal{V}_r$}. Without loss of generality, we consider a subsystem composed of misbehaving node $i$ and its neighbors. 
\CR{Since $W$ is doubly stochastic, we have the following according to (\ref{e2}):}
\begin{equation}\label{proof0.2}
\sum_{j \in \mathcal{V}} x_j(k+1)=\sum_{j \in \mathcal{V}} x_j(0)+ \sum_{l=1}^{k} \sum_{j \in \mathcal{V}}\varepsilon_j(l).
\end{equation}

\WZ{\textbf{Case 1:}} If node $i$ is not isolated, according to Compensation Scheme \RNum{1} and \RNum{2}, we have
\begin{equation}\label{proof0.1}
\sum_{l=0}^{ \infty } \{  \varepsilon_i(l) + \sum_{j \in \mathcal{N}_i} (\eta_{j}^{i(1)}(l+1)+\eta_{j}^{i(2)}(l+1) ) \} =0.
\end{equation}
Hence, (\ref{th1}) holds.

\WZ{\textbf{Case 2:}} If node $i$ is isolated at time $k+1$, we can regard it as staying at the state $x_i(k+1)$. 
Since $W$ is a doubly stochastic matrix at each-step, we have that (\ref{proof0.2}) holds.
 Node $i$ is regarded to stay at the value $x_i(k+1)$, i.e. $\varepsilon_i(l)=0,\forall l > k$. \CR{Because Compensation Schemes \RNum{1} and \RNum{2} will compensate the impact before isolation, the condition (\ref{proof0.1}) holds.}
In addition, Compensation Scheme \RNum{3} will compensate \WZ{the influence of node $i$ on the system before isolation.} 
Thus, we have
\begin{equation}\label{proof2}
\sum_{l=1}^{ \infty }  \sum_{j \in \mathcal{V}} \varepsilon_j(l)= x_i(k+1)-x_i(0).
\end{equation}
Combining  (\ref{proof0.2}) with (\ref{proof2}) and noting that node $i$ is regarded as staying at the value $x_i(k)$ after isolation, it follows that
\begin{equation}\label{proof4}
\lim_{l \to \infty }\sum_{j \in \mathcal{V}/ \{i\}} x_j(l) =  \sum_{j \in \WZ{\mathcal{V}/\{i\}}} x_j(0).
\end{equation}
Generally, we have  
\begin{equation}\label{proof5}
\lim_{l \to \infty }\sum_{j \in \mathcal{V}_r} x_j(l) =  \sum_{j\in \mathcal{V}_r} x_j(0).
\end{equation}
Hence, the \WZ{resilient} average consensus \CR{among the nodes in the set after isolation $\mathcal{V}_r$} is achieved.
\end{proof}

\begin{remark}
Theorem \ref{th1.1} guarantees the accurate average consensus even if there are misbehaving nodes in the system. 
Due to the use of local error bound, D-DCC provides tolerance for accidental miscalculations and transmission errors.
\DI{If isolation is not adopted in these scenarios, e.g., a node malfunctions transitorily, \WZ{the faulty node is still able to accomplish the system mission,} which may contribute to the overall efficiency of the system.} 
The tolerance depends on the parameters $\alpha_i$ and $\rho_i$.
Note that a larger $\alpha_i\rho_i$ will improve the fault tolerance but reduce the convergence speed.
\end{remark}

\subsection{Analysis of S-DCC}
\CR{Before going into the analysis of S-DCC algorithm,} \WZ{we first define several notations.
Summing up compensation input of all neighbors of \WZ{misbehaving} node $i$}, the average compensation of $\varepsilon_i(k)$ is $ -\sum_{j \in \mathcal{N}_i}\overline{\varepsilon}_{i}^{j}(k) \triangleq -\overline{\varepsilon}_{i}(k)$, where  $\overline{\varepsilon}_{i}(k)$ is the average error of detecting times.
Note that there may be some faulty nodes that will not be isolated. 
These faulty nodes may cause errors due to accidental miscalculations.
Hence, it is reasonable to assume that the errors of \CR{the faulty node $i$} occur in a period from $k_i^0$ to $k_i^1$.
Denote subsets $\mathcal V_f$ and $\mathcal V_M$  be the set of faulty nodes and malicious nodes, respectively. 
We provide the following theorem to analyze the performance of S-DCC.
\begin{theorem}\label{thm:sdca}
If Assumptions  1-3  hold, then S-DCC achieves \WZ{resilient} average consensus in expectation \WZ{among the nodes in the set after isolation $\mathcal{V}_r$}, i.e., $ \forall j\in \mathcal{V}_r$, \CR{we have}
\begin{equation} \label{th2.11}
\mathbb{E} \big [ \lim_{l\to \infty} x_j(l)\big ] = \frac{1}{|\mathcal{V}_r|} \sum_{u\in \mathcal{V}_r} x_u(0),
\end{equation}
\begin{equation} \label{th2.2}
 \mathbb{D}\big [ \lim_{l\to \infty} x_j(l)\big ] \le \sum_{i\in \mathcal{V}_m} D_m + \sum_{i\in \mathcal{V}_f} D_f ,
\end{equation} 
where 
\begin{equation*}
\begin{aligned}
~& D_m = {{(k^\text{iso}_i-M_i)}(1+\frac{k^\text{iso}_i-M_i}{M_i})\sigma_{\varepsilon_i}^2},\\
~& D_f \!= \! \frac{1-p}{p^2}(\frac{\sigma_{\varepsilon_i}^2}{M_i}+\theta_i^2\mu_i^2)
\!+\!{{(k^{1}_i-k^{0}_i)}(1+\frac{k^{1}_i-k^{0}_i}{M_i})\sigma_{\varepsilon_i}^2},
\end{aligned}
\end{equation*}
in which $k^\text{iso}_i$ is the isolation time of node $i$ and $M_i = \min_{j\in \mathcal{N}_i} m_j(k^\text{iso}_i)$. 
In addition, the consensus value is bounded, i.e.,
\begin{equation}\label{th2.3}
\big| \lim_{l\to \infty} x_j(l)- \frac{1}{|\mathcal{V}_r|} \sum_{u\in \mathcal{V}_r} x_u(0) \big| \le \displaystyle{ \frac {\alpha \rho|\mathcal{V}_m|} {(1-\rho)|\mathcal{V}_r|}}. 
\end{equation}
\end{theorem}

\begin{proof}
First, we illustrate that all \CR{malicious} nodes will be detected with probability one.
For each \CR{malicious} node $i$, it \CR{affects} the system with probability $\theta_i,0 < \theta_i \le 1$. \DI{For each normal node $j\in\mathcal{N}_i$, it  detects the misbehavior of node $i$ with  probability $p,0<p \le 1$.} The probability of the event that node $i$ is detected by $j$ in no later than time $k$ is
\begin{equation}\label{pr01}
P(k)=1-(1-p\theta_i)^k.
\end{equation}
By taking the limit on both sides of (\ref{pr01}), we have
\begin{equation*}\label{pr02}
\lim_{k \to \infty} P(k)=\lim_{k \to \infty} (1-(1-p\theta_i)^k) =1.
\end{equation*}
Hence, all \CR{malicious} nodes will be detected with probability one.
\WZ{Similarly, all \CR{malicious} nodes will be \CR{isolated} with probability one because the bound $\alpha \rho ^k = 0$  as $k \to \infty$.}

Next, we will \CR{prove that}
\begin{equation} \label{sdcapr1}
\mathbb{E} \big[ \lim_{l \to \infty }\sum_{j \in \mathcal{V}_r} x_j(l) \big] =  \sum_{j \in \mathcal{V}_r} x_j(0).
\end{equation}
\WZ{For the sake of simplicity, we perform analysis on a subsystem composed of the misbehaving node $i$ and its neighbors in $\mathcal{N}_i$ in the following.} 
Note that $X_i(l)$  and $Y_i(l)$ are independent.

\WZ{Case 1}: Considering a \CR{malicious} node $i$, the expectation of the sum of its error within time $k$ is
\begin{equation}\label{th4}
\mathbb{E} \left[ \sum_{l=1}^{k}\varepsilon_i(l) \right]=\mathbb{E} \left[ \sum_{l=1}^{k}X_i(l)Y_i(l)\right] = k\theta_i\mu_i.
\end{equation}

Without loss of generality, \CR{we consider that all the neighbor nodes in $\mathcal{N}_i$ detect the misbehavior of node $i$ at the same time.}
The expectation of $\sum_{j\in \mathcal{N}_i} \overline{\varepsilon}_{i}^{j}$ satisfies
\begin{equation*}
\begin{aligned}
\mathbb{E}\Big[ \sum_{j \in \mathcal{N}_i} \overline{\varepsilon}_{i}^{j}(k) \Big] ~&=\mathbb{E}\Big[ \frac{1}{m_j}{\sum_{j \in \mathcal{N}_i} \sum_{\varepsilon_{i}^j(k) \in  \Omega_j^{(i)}} \varepsilon_{i}^j(k)} \Big]\\
~&= \mathbb{E}\big[\varepsilon_i(l) \big]
=\theta_i\mu_i.
\end{aligned}
\end{equation*}
The expectation of the Compensation Scheme \RNum{1}, \RNum{2}, \RNum{4} is
\begin{equation}\label{th5}
\begin{aligned}
&\mathbb{E} \left[ \sum_{j \in \mathcal{N}_i} \sum_{l=1}^{k} (\eta_{ji}^{(1)}(l)+\eta_{ji}^{(2)}(l)+\eta_{ji}^{(4)}(l))\right] \\
=&\mathbb{E} \left[ {\sum_{j \in \mathcal{N}_i} \sum_{\varepsilon_{i}^j(k) \in  \Omega_j^{(i)}} \varepsilon_{i}^j(k)}-(k-m_j(k))\sum_{j \in \mathcal{N}_i}\overline{\varepsilon}_{i}^{j} \right] \\
=&\mathbb{E} \left[ -k\sum_{j \in \mathcal{N}_i} \overline{\varepsilon}_{i}^{j} (k)\right] =-k\theta_i\mu_i.
\end{aligned}
\end{equation}
Let $k+1=k^\text{iso}_i$ for (\ref{com3}).
Therefore, combining (\ref{com3}), (\ref{th4}) and (\ref{th5}), we have
\begin{equation*}
\begin{aligned}
\mathbb{E}\left[ \sum_{l=1}^{k}(\varepsilon_{i}(l)+ \sum_{j\in \mathcal{N}_i} \eta_j(l))\right]=x_i(k+1)-x_i(0).
\end{aligned}
\end{equation*}
Hence, we have 
\begin{equation}\label{proofSDCC1}
\mathbb{E} \left[ \lim_{l \to \infty }\sum_{j \in \mathcal{V}/ \{i\}} x_j(l)\right]=\sum_{j \in \WZ{\mathcal{V}/\{i\}}} x_j(0).
\end{equation}
\WZ{Case 2}: Consider that the errors of faulty node $i$ which occur in a period from $k_i^0$ to $k_i^1$. The compensation for node $i$ is $ -\sum_{j \in \mathcal{N}_i} (k_i^{j0}-k_i^{j1}) \overline{\varepsilon}_{i}^{j}$, where $k_i^{j0}$ is the last time of detection before $k_i^0$ and $k_i^{j1}$ is the last time of detection before $k_i^1$. Consequently, the following holds: 
\begin{equation*}
\mathbb{E}[ k_i^0-k_i^{j0}]=\mathbb{E}[k_i^1-k_i^{j1} ]=\frac{1}{p}-1.
\end{equation*}
Hence, we have $\mathbb{E} [k_i^{j0}-k_i^{j1}]=k_i^{0}-k_i^{1}$.
Since detection and errors are independent, it holds that
\begin{equation*}
\mathbb{E}\left [-\sum_{j \in \mathcal{N}_i} (k_i^{j0}-k_i^{j1}) \overline{\varepsilon}_{i}^{j} \right]=(k_i^1-k_i^0)\theta_i\mu_i. 
\end{equation*}
Then, we have 
\begin{equation}\label{proofSDCC2}
\mathbb{E} \left[ \lim_{l \to \infty }\sum_{j \in \mathcal{V}} x_j(l) \right] =  \sum_{j \in \WZ{\mathcal{V}}} x_j(0).
\end{equation}
\WZ{With the above two cases, we have (\ref{sdcapr1}) for the general set $\mathcal{V}_r$.}
Hence, (\ref{th2.1}) holds and an average consensus in expectation \WZ{among the nodes in the set after isolation $\mathcal{V}_r$} is achieved. 


Since $\overline{\varepsilon}_{i}$ is the average value of sampling, we have $\mathbb{D} [ \overline{\varepsilon}_{i}] \le \sigma_{\varepsilon_i}^2/{M_i}$.
Because the detected errors ($m_j$ times) will be compensated accurately by node $j$, the variance of the consensus value is given by
\begin{equation*}
\begin{aligned}
&\mathbb{D} \Big[\sum_{u\in \mathcal{V}_r} x_u(k)\Big]\\
=& \sum_{i\in \mathcal{V}_m}\mathbb{D}\Big[{\sum_{l=1}^{k^\text{iso}_i}\varepsilon_i(l)} -k^\text{iso}_i\sum_{j \in \mathcal{N}_i} \overline{\varepsilon}_{i}^{j} \Big] \\
+& \sum_{i \in \mathcal{V}_f} \mathbb{D} \Big[ {\sum_{l=k_i^0}^{k_i^1}}\varepsilon_i(l)-\sum_{j \in \mathcal{N}_i} (k_i^{j0}-k_i^{j1}) \overline{\varepsilon}_{i}^{j} \Big] \\
\le& \sum_{i\in \mathcal{V}_m} {{(k^\text{iso}_i-M_i)}\big(1+(k^\text{iso}_i-M_i)/M_i\big)\sigma_{\varepsilon_i}^2}\\
+&\sum_{i\in\mathcal{V}_f} \Big[\frac{1-p}{p^2}(\frac{\sigma_{\varepsilon_i}^2}{M_i}+\theta_i^2\mu_i^2)
\!+\!{{(k^{1}_i-k^{0}_i)}(1\!+\!\frac{k^{1}_i-k^{0}_i}{M_i})\sigma_{\varepsilon_i}^2}\Big]\\
=& \sum_{i\in \mathcal{V}_m} D_m + \sum_{i\in \mathcal{V}_f} D_f. 
\end{aligned}
\end{equation*}
Each malicious node causes \CR{adverse impacts}. Hence, (\ref{th2.2}) holds.

According to the proof of Theorem \ref{th1.1}, $\varepsilon(k)$ satisfies the condition $||\varepsilon(k)||_\infty \le \alpha \rho^k$.
For misbehaving nodes, we have
\begin{equation*}
\sum_{i\in \mathcal{V}_m}\sum_{l=1}^{\infty} \varepsilon_i(l) \le \frac {|\mathcal{V}_m|\alpha \rho} {(1-\rho)}.
\end{equation*}
Hence, (\ref{th2.3}) is proved.
\end{proof}
\begin{remark}
Theorem \ref{thm:sdca} ensures \WZ{resilient} average consensus in expectation.
Note that if $M_i=k^\textrm{iso}_i$, then $\mathbb{D} \big[\lim_{l\to \infty} x_j(l)\big]=0$, which corresponds to deterministic conditions by D-DCC.
It can be seen that misbehaving nodes can be detected in finite time in expectation.
Meanwhile, the expectation of undetected errors is the same as the mean detected errors, and for faulty nodes, compensation period $k_i^{j0}-k_i^{j0}$ has the same expectation as that of errors, i.e., $k_i^1-k_i^0$.
Hence, misbehavior can be compensated in expectation. 
The first part of variance $D_m$ is from undetected errors and compensation of them for both malicious nodes and faulty nodes. The other part $D_f$ is because of the difference between compensation period $k_i^{j0}-k_i^{j0}$ and error period from $k_i^1$ to $k_i^0$ for faulty nodes.

For normal nodes, the larger attack probability $\theta_i$ of neighboring \WZ{misbehaving} nodes and detection probability $p$ will decrease the expected steps for detection. On the one hand, a larger attack probability will improve the attack capability of malicious nodes. On the other hand, the detection probability based on the reliable link will be close to one. Hence, the performance of detection will be improved.
Although the variance may be large, it mostly depends on the undetected attack errors.
As for the bias between undetected errors and compensation,  unpredictable attacks with large variance may cause considerable bias to our algorithm, but it will be out of error bound easily.
Furthermore, it is not actually necessary for malicious nodes to manipulate the information set with a constant probability and a certain distribution. The detection method will be effective as long as the attack probability is larger than zero, and the errors may follow a certain attack method. Hence, to simplify the statement, we assume a constant attack probability and present the attack errors by a time-invariant probabilistic model.
\end{remark}

Next, we analyze the accuracy of mean-based Compensation Scheme \RNum{4}, i.e., the distance between the mean-based compensation and actual errors. 
The actual errors may consist of multiple uncertainties.
Let $F_{Y_i(k)}(x)$ be the cumulative distribution function (CDF) of $Y_i(k)$.
We adopt Gaussian mixture model (GMM) to represent the error variable $Y_i(k)$.
On the one hand, any distribution can be generally modeled by GMM with arbitrary precision. 
On the other hand, GMM has good operation properties.
GMM is defined as a convex combination of Gaussian distribution with different expectations $\mu_l$ and variances $\sigma_l$:
\begin{equation}
    F_{Y_i(k)}(x)=\sum_{l=1}^{N_l} a_l \Phi(\frac{x-\mu_l}{\sigma_l}), \sum_{l=1}^{N_l} a_l=1,\sum_{l=1}^{N_l} a_l\mu_l=\mu_i ,
\end{equation}
where $\Phi(\cdot) $ is the CDF of the standard normal distribution.
Since $\varepsilon_i(k) = X_i(k)Y_i(k)$, the CDF of $\varepsilon_i(k)$ is
\begin{equation} \label{distribution}
F_{\varepsilon_i(k)}(x)=
\begin{cases}
\theta_iF_{Y_i(k)}(x) & x<0\\
1-\theta_i+\theta_iF_{Y_i(k)}(x) & x\ge 0.
\end{cases}
\end{equation}
Without loss of generality, consider all the detection numbers $m_j(k), j\in\mathcal{N}_i $ are the same. 
When detecting time is large enough, according to the Central Limit Theorem~\cite{grimmett2003probability}, we have 
\begin{equation*}
\overline{\varepsilon}_{i} \sim \mathcal{N}(\theta_i\mu_i,\sigma_{\varepsilon_i}^2/M_i ),
\end{equation*}
\WZ{where $\mathcal{N}(\cdot,\cdot)$ is the Gaussian distribution.}
The CDF of $\overline{\varepsilon}_{i}$ is
\begin{equation*}
F_{\overline{\varepsilon}_{i}}(x)= \Phi(\frac{\sqrt{M_i}(x-\theta_i\mu_i)}{\sigma_{\varepsilon_i}})
\end{equation*}

The close proximity of error and compensation probability distributions will not only guarantee the close final consensus value, but also guarantee the stationarity of consensus process. 
The characteristic of proximity of two probability distributions can be described by the Wasserstein distance\cite{vallender1974calculation}. 
The Wasserstein distance $R(\mathcal{P},\mathcal{Q})$ between the two distributions $\mathcal{P}$ and $\mathcal{Q}$ is defined as follows:
\begin{equation*}
R(\mathcal{P},\mathcal{Q}) =\inf \mathbb{E} \big[\rho (\xi,\eta) \big],
\end{equation*}
where \CR{$\rho(\cdot,\cdot)$ is the function} of the metric space and the mathematical operation $\inf$ is taken over all possible pairs of random variables $\xi$ and $\eta$ with distributions $\mathcal{P}$ and $\mathcal{Q}$, respectively. 
In the case of one-dimensional space with the Euclidean metric, the Wasserstein distance is calculated by
\begin{equation*}
R(\mathcal{P},\mathcal{Q}) = \int_{-\infty}^{\infty} |F(x)-G(x)|\text{d}x,
\end{equation*}
where \WZ{$F(x)$ and $Q(x)$ are the CDF of $\mathcal{P}$ and $\mathcal{Q}$,} respectively\cite{vallender1974calculation}.
Hence, we have the Wasserstein distance between $\varepsilon_i(k)$ and $\overline{\varepsilon}_i$:
\begin{equation}
R(\varepsilon_i(k),\overline{\varepsilon}_i) = \int_{-\infty}^{\infty} |F_{\varepsilon_i(k)}(x)-F_{\overline{\varepsilon}_i}(x)|\text{d}x.
\end{equation}
We can use the Wasserstein distance to show the expectation of absolute error between the mean-based compensation and actual errors.
We provide the following theorem to illustrate the bound of $R(\varepsilon_i(k),\overline{\varepsilon}_i)$ when $Y_i(k)$ obeys the normal distribution.
\begin{table*}[htbp]
    \centering
    \caption{Comparisons of the proposed methods with other representative algorithms}
    \begin{tabular}{lcccccc}
    \hline
    \hspace{-2mm}
Method\! & Topology\!\! & \!\!Attack Model\!\! & Nodes & \!Information Needed \!& \!\!Consensus Value\!\!\\
\hline 
\hspace{-2mm}W-MSR\cite{6315661}\! & Directed\!\!  & \!\!Deception\!\! & Byzantine & \!Local Information\! & \!\!Convex Hull\!\!\\
\hspace{-2mm}SDA \& MDA\cite{8013830}\!\! & Undirected\!\! & \!\!Deception\!\! & Malicious & \!Two-hop Information\! & \!\!Convex Hull\!\!\\
\hspace{-2mm}Algorithm in \cite{yuan2021secure}\!\! & Directed\!\! & \!\!Deception\!\! & Malicious & \!Two-hop Information\! &  \!\! - \!\!\\
\hspace{-2mm}Algorithm in \cite{Pasqualetti201290}\!\! & Directed\!\! & \!\!Deception\!\! & Byzantine & \!Partial Global Information\! & \!\!- \!\!\\
\hspace{-2mm}D-DCC\! & Undirected\!\! & \!\!Deception\!\! & Malicious & \!Two-hop Information\! & \!\!Average Consensus\!\!\\
\hspace{-2mm}S-DCC\! & Undirected\!\! & \!\!\small{Deception \!\!\&\!\! Link Failure}\!\! & Malicious & \!Two-hop Information\! &\!\! Average in Expectation\!\!\\
\hline
\end{tabular}
\label{tab:compare}
\end{table*}
\begin{theorem}\label{aaa}
When $Y_i(k)$ is modeled by GMM, we have
\begin{equation}\label{wasserstein}
R(\varepsilon_i(k),\overline{\varepsilon}_i) \!\!\le\!\! {(1\! \!-\!\!\theta_i)\mathbb{E} [|Y_i|] \!\!+\!\!\! \sum_{l=1}^{N_l} \!a_l(|\theta_i\mu_i \!-\! \mu_l| \!\!+\!\!|\frac{\sigma_{\varepsilon_i}}{\sqrt{M_i} \!-\! \sigma_l}|)},
\end{equation}
where $\mathbb{E} [|Y_i|] \!\le\! \sum_{l=1}^{N_l}\! \{\sqrt{\frac{2}{\pi}}\sigma_l\exp{(\frac{-\mu_l^2}{2\sigma_l^2})} \!+\!\mu_l[1\!-\!2\Phi(\frac{-\mu_l}{\sigma_l})]\}$.
\end{theorem}
\begin{proof}
Referring to the absolute value inequality, we have
\begin{equation*}
\begin{aligned}
~&R(\varepsilon_i(k),\overline{\varepsilon}_i) = \int_{-\infty}^{\infty} |F_{\varepsilon_i(k)}(x)-F_{\overline{\varepsilon}_i}(x)|\text{d}x\\
~&\le\!\! \int_{-\infty}^{\infty}\!\! |F_{\varepsilon_i(k)}(x)-F_{Y_i}(x)|\text{d}x 
+\!\! \int_{-\infty}^{\infty} |F_{\overline{\varepsilon}_i}(x)-F_{Y_i}(x)|\text{d}x.
\end{aligned}
\end{equation*}
According to (\ref{distribution}), we have 
\begin{equation*}
\begin{aligned}
~&\!\int_{-\infty}^{\infty} |F_{\varepsilon_i(k)}(x)-F_{Y_i}(x)|\text{d}x \\
~&\!\!\!\!=\!\!\int_{-\infty}^{0}|(1-\theta_i)F_{Y_i}(x)|\text{d}x+\!\! \int_{0}^{\infty}|(1-\theta_i)(1-F_{Y_i}(x))|\text{d}x\\\!
~&\!\!\!\!=(1-\theta_i)\left[ \int_{-\infty}^{0}F_{Y_i}(x)\text{d}x+\int_{0}^{\infty}(1-F_{Y_i}(x))\text{d}x \right]\\
~&\!\!\!\! =(1-\theta_i)\mathbb{E} [|Y_i|].
\end{aligned}
\end{equation*}

In addition, according to the Wasserstein distance between two normal distributions~\cite{chafai2010fine}, we have
\begin{equation*}
\begin{aligned}
~&\int_{-\infty}^{\infty} |F_{\overline{\varepsilon}_i}(x)-F_{Y_i}(x)|\text{d}x \\
~&\le
\sum_{l=1}^{N_l} a_l \left[
\int_{-\infty}^{\infty}|F_{\overline{\varepsilon}_i}(x)-\Phi(\frac{x-\mu_l}{\sigma_l})|\text{d}x
\right]\\
~&\le \sum_{l=1}^{N_l}  a_l(|\theta_i\mu_i-\mu_l|+|\sigma_{\varepsilon_i}/\sqrt{M_i}-\sigma_l|)
\end{aligned}
\end{equation*}
Combining the above, we complete the proof.
\end{proof}
\begin{remark}
Theorem \ref{aaa} shows the Wasserstein distance under GMM. Hence, the expectation of absolute error between the mean-based compensation and actual errors is bounded.
 Generally, the explicit bound is difficult to formulate under other distributions, but the Wasserstein distance is bounded as long as $\mathbb{E}[|Y_i|]$ and $\mathbb{E}[|\overline{\varepsilon}_i|]$ exist.
The quantitative evaluations can be found in Sec.~\ref{SDCA5b}.
\end{remark}
\begin{remark}
\CR{Here, we provide a comprehensive comparison of our D-DCC and S-DCC algorithms with other resilient consensus algorithms, which are summarized in Table \ref{tab:compare}.}
\end{remark}

\section{Numerical Evaluations}
In this section, we conduct numerical evaluations to illustrate the performance of D-DCC and S-DCC. 
Consider a Erdös-Rényi Random graph (where the probability for edge creation is $0.7$) with $N=10$ nodes. The system updates states by (\ref{e2}), where $W$ is designed by Perron weights.
All nodes' \CR{initial} states are selected from the interval $[0,2]$ randomly.
In the network, there are two misbehaving nodes that are not neighbors, i.e., malicious node $1$ that intends to break average consensus and faulty node $5$.
We set $\alpha_i=5$, $\rho_i=0.9$, $\forall i\in \mathcal{V}$.

\subsection{Performance of D-DCC}
At this part, we set the \CR{adverse impacts} of nodes $1$ and $5$ to satisfy $\varepsilon_1(k)=0.5 \cos (k)$ and $\varepsilon_5(k)= 0.5\times0.6^k$, \CR{respectively}.
Fig.~\ref{fig1.a} shows that all nodes except node $1$ achieves consensus. The consensus value is the average value of initial states of remaining nodes showing as the blue dotted line. 
As a contract, we plot the average state of normal nodes by MSR algorithm~\cite{262588} as the green line.
 The exact \WZ{resilient} average consensus is achieved by D-DCC, but MSR algorithm does not guarantee average consensus.
Fig.~\ref{fig1.b} shows the errors of node $1$ and $5$. Node $5$ has not been isolated because its error is exponentially decaying and in the local bound.
Node \CR{$1$} is isolated at time $24$ because the error is out of the bound.

\begin{figure}[t]
\begin{center}
\includegraphics[height=4cm]{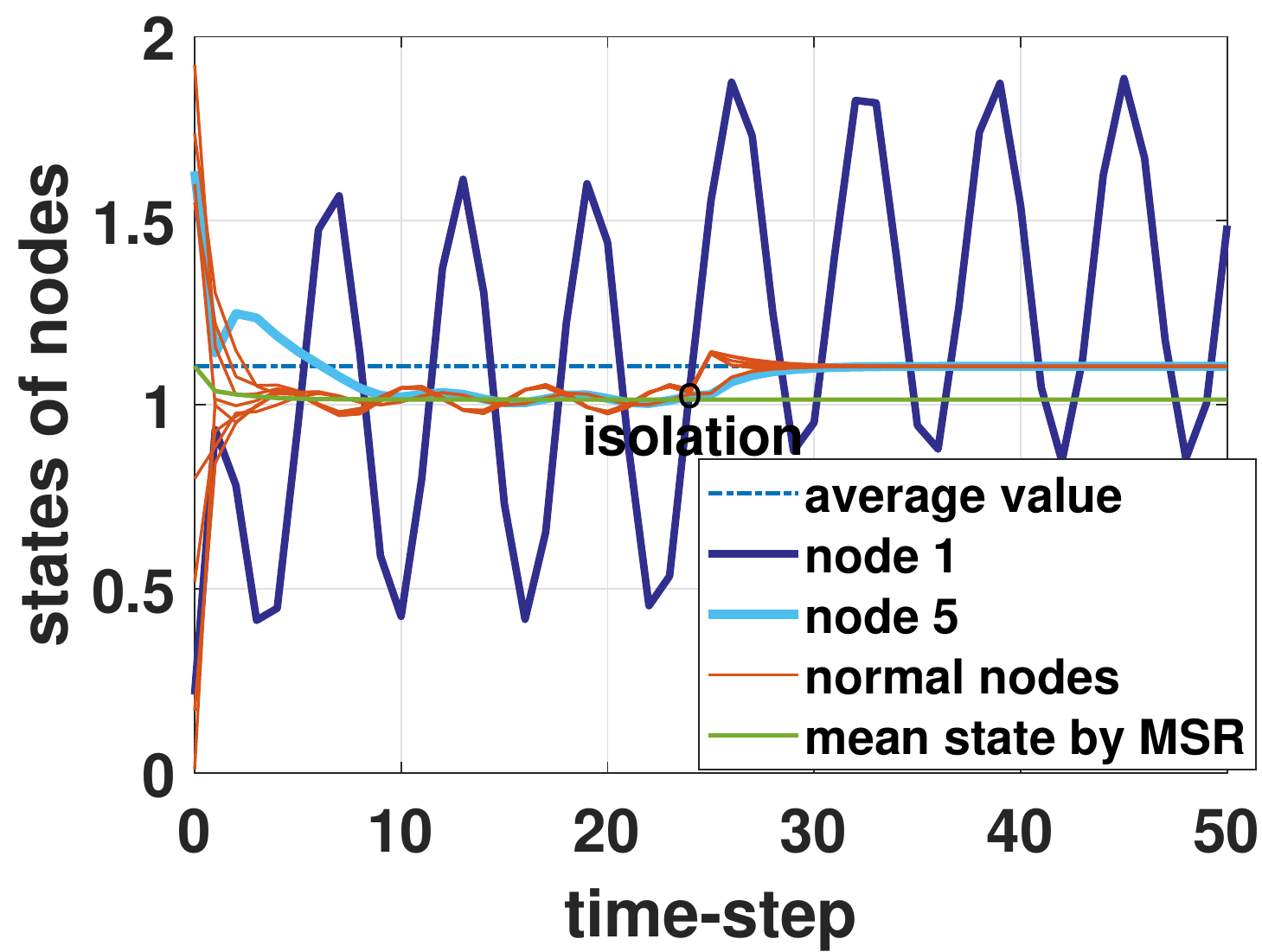}    
\caption{Performance of D-DCC: States of nodes.}  
\label{fig1.a}                                 
\end{center}                                 
\end{figure}

\begin{figure}[t]
\begin{center}
\includegraphics[height=4cm]{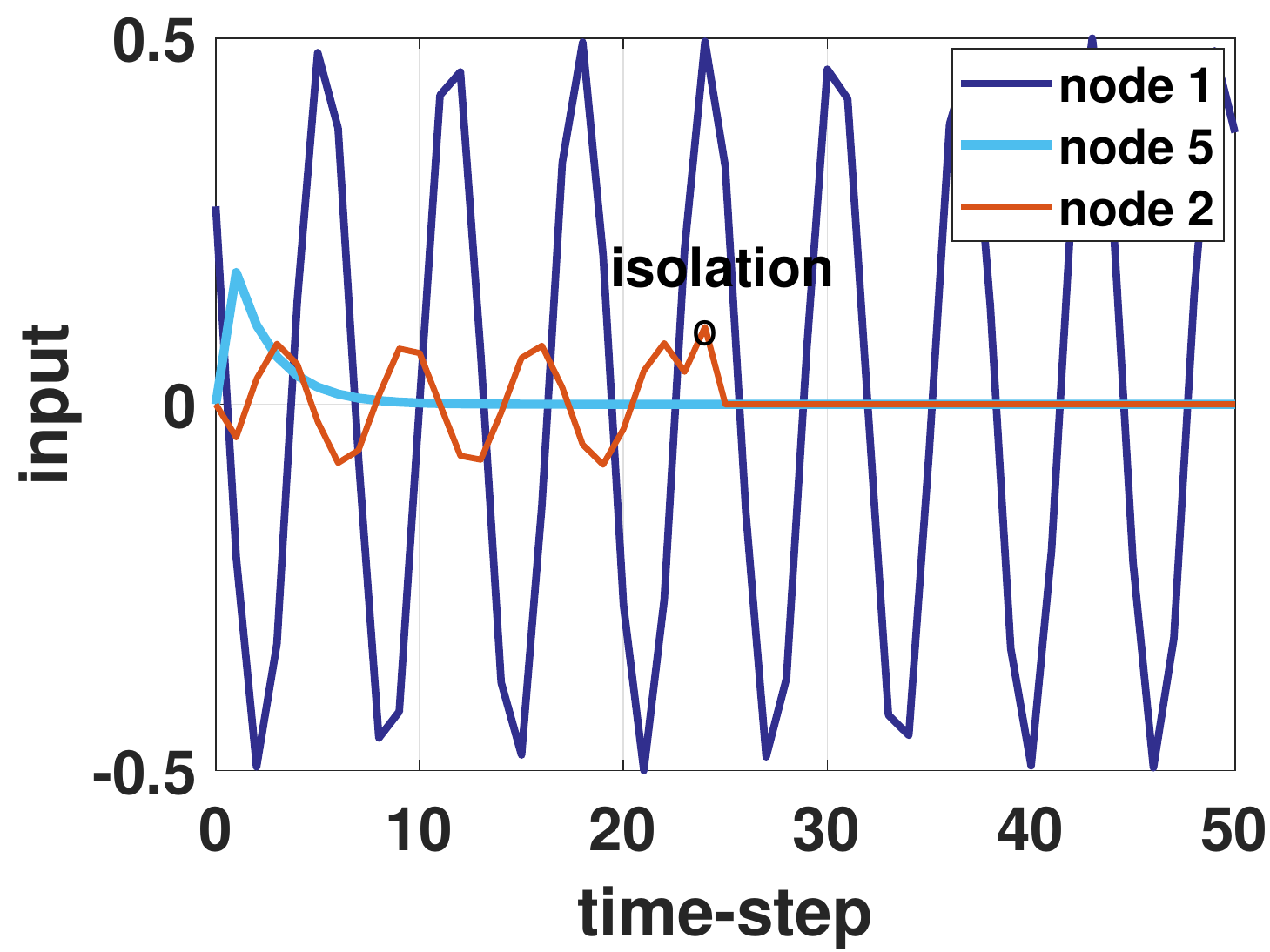}    
\caption{Performance of D-DCC: Errors of misbehaving nodes ($1$ and $5$) and compensation input of node $2$.}  
\label{fig1.b}                                 
\end{center}                                 
\end{figure}

\subsection{Performance of S-DCC}\label{SDCA5b}
The error of malicious node $1$ is set to obey a GMM \WZ{when attack is adopted}, i.e.,  $ F_{Y_1(k)}(x)= 0.5 \Phi(\frac{x-0.05}{\sqrt{0.05}}) + 0.5 \Phi(\frac{x-0.15}{\sqrt{0.2}})$.  The error of faulty node $5$ follows a normal distribution. We set the connection probability $p=0.8$ and attack probability $\theta_1=0.8$.

Fig.~\ref{fig2.a} shows that all nodes except node $1$ achieve consensus, and node $1$ is isolated by other nodes but node $5$ is not.
Node $1$ causes errors continuously, and it is isolated at time $36$ when the error is out of the error bound.
Node $5$ only makes misbehavior in first $10$ steps, and the errors are within the fault-tolerance bound. Hence, node $5$ is not isolated, and its errors are compensated by neighbors. 
The consensus value is close to the average value of initial states of remaining nodes. 
Though the limit value is the exact average consensus in expectation, in practice it may vary from it. 
Compared with MSR algorithm~\cite{262588}, S-DCC achieves more accurate average consensus.
Fig.~\ref{fig2.b} shows the compensation results of node $2$ and the errors of nodes $1$ and $5$.

According to the Wasserstein distance, we have $R(\varepsilon_1(k),\overline{\varepsilon}_1) = 0.1245$, which is much smaller than the bound presented in Theorem $3$.
The CDF of  $\varepsilon_1(k)$ and $\overline{\varepsilon}_1$ are shown in Fig. 3.
 The expectation of absolute error between the \CR{injected adverse impact} and the mean-based compensation at each time is within a small range.
The two distributions are similar, which shows the accuracy of mean-based compensation. 

In the same scenario, we repeat the test for 1000 times with the same initial states of all nodes to validate the correctness. 
\WZ{The average initial states of nodes $\{2,\ldots,10\}$ is 1.1517. 
The results are shown in Table \ref{tab2}.}
It can be observed that the average consensus value of S-DCC is closer to average initial states than that of MSR algorithm.
The variance of consensus value by S-DCC is $0.0166$, which is much smaller than the bound. Because the upper bound of variance is relaxed. With all above results, the effectiveness of our proposed D-DCC and S-DCC is illustrated.


\begin{figure}[t]
\begin{center}
\includegraphics[height=4cm]{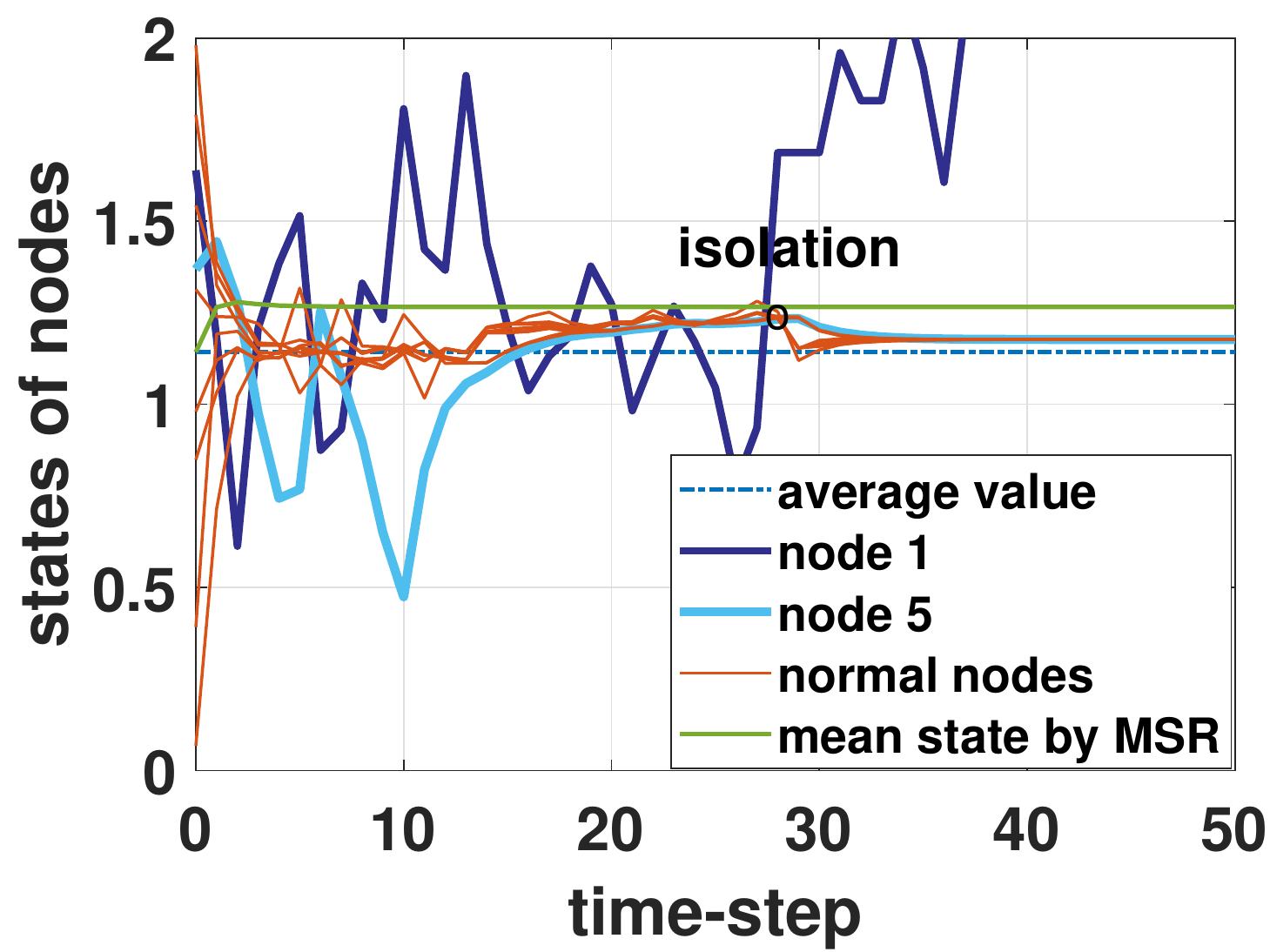}    
\caption{Performance of S-DCC: States of nodes.}  
\label{fig2.a}                                 
\end{center}                                 
\end{figure}

\begin{figure}[htbp]
\begin{center}
\includegraphics[height=4cm]{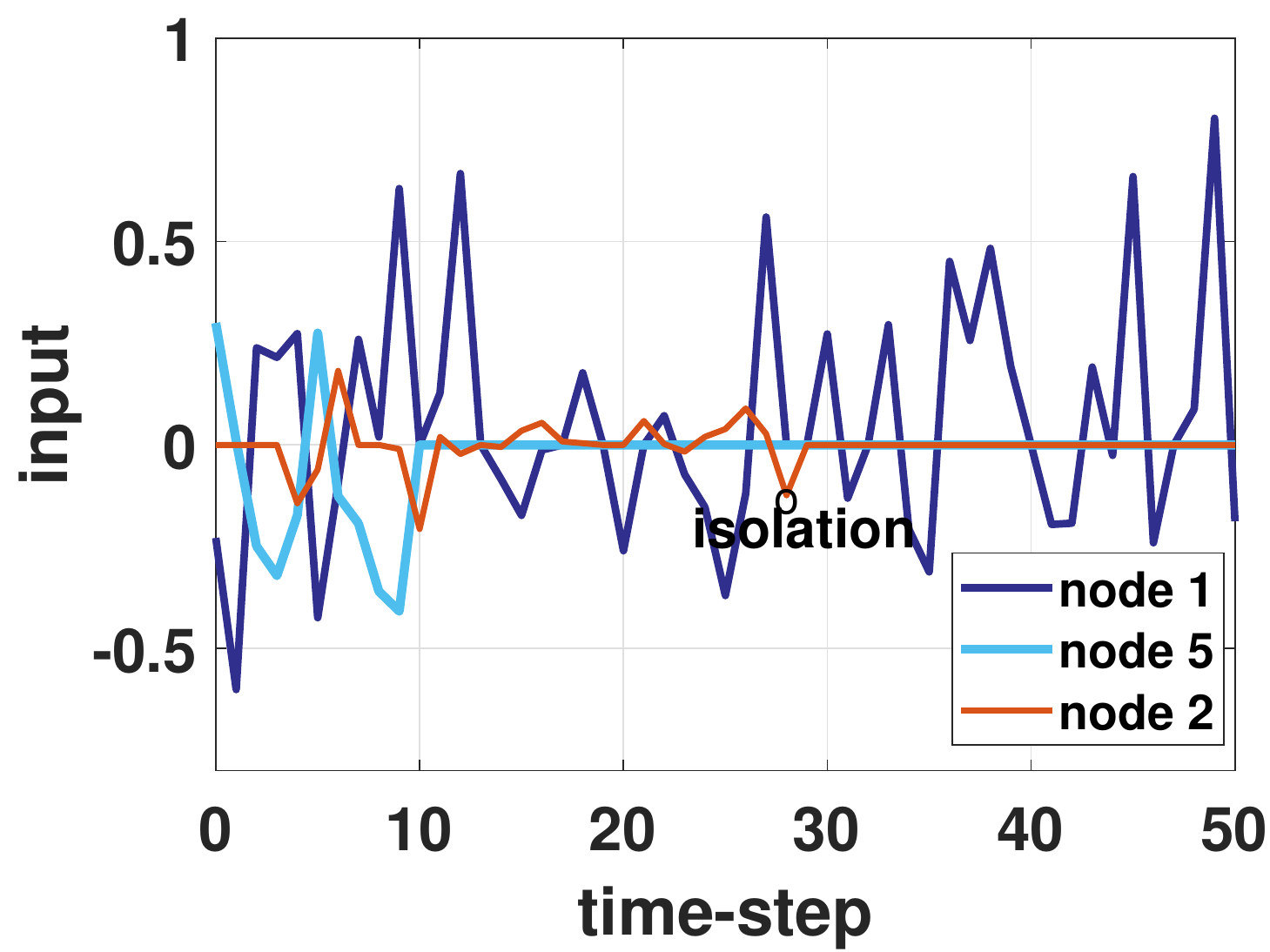}    
\caption{Performance of S-DCC: Errors of misbehaving nodes ($1$ and $5$) and compensation input of node $2$.}  
\label{fig2.b}                                 
\end{center}                                 
\end{figure}

\begin{figure}[htbp]
\begin{center}
\label{fig4}
\centering\includegraphics[width=0.6\columnwidth]{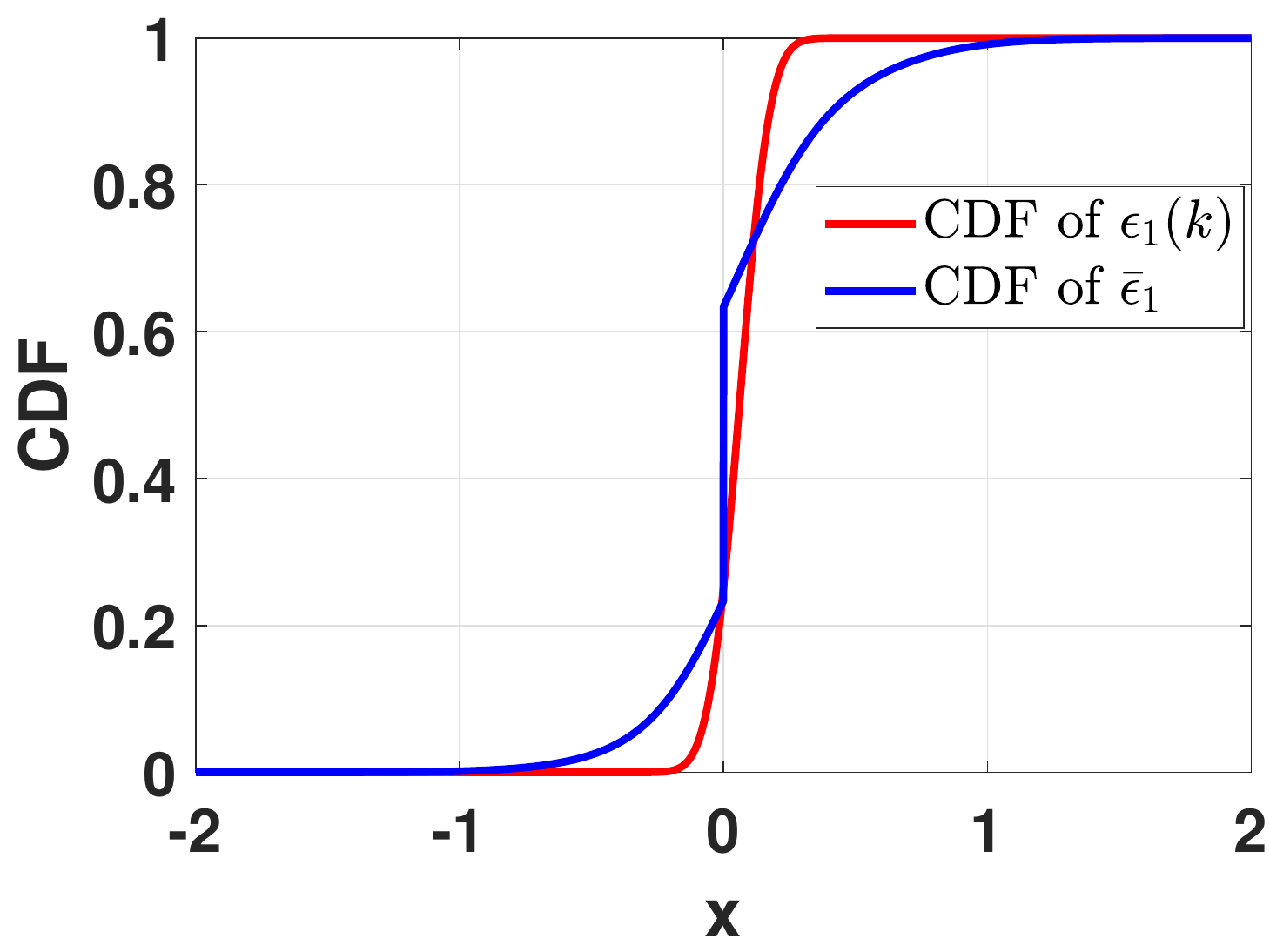}
\caption{ CDF of  $\varepsilon_1(k)$ and $\overline{\varepsilon}_1$}
\end{center}
\end{figure}

\begin{table}[t]
    \centering
    \caption{Comparisons of the proposed methods with MSR algorithm}
    \begin{tabular}{cc}
    \hline
    Methods &  Mean Consensus Value (1000 times)\\
    \hline
    D-DCC & 1.1517\\
    S-DCC & 1.1330 \\
    MSR algorithm & 1.0841\\
    \hline
    \end{tabular}
    \label{tab2}
\end{table}
\section{CONCLUSION}

In this paper, we have investigated the resilient average consensus problem against misbehaving nodes in multi-agent systems. 
We have first presented the D-DCC algorithm to compensate the \CR{adverse impacts} caused by misbehaving nodes while considering reliable communication.
The exponential decaying bound provides fault tolerance for misbehaving nodes and guarantee the convergence.
We have proved that the \WZ{resilient} average consensus can be achieved by D-DCC.
Furthermore, we have proposed S-DCC algorithm with mean-based compensation to adapt for scenarios where link failures may occur.
It \WZ{has been} proved that the  \WZ{resilient} average consensus in expectation is achieved by S-DCC, and the absolute error between mean-based compensation and actual \CR{adverse impact} has been analyzed by the Wasserstein distance.
Finally, simulations \WZ{have been} conducted to illustrate the effectiveness of the proposed algorithms.

There are still many issues worthy of further investigations. 
First, achieving exact resilient consensus over time-varying and directed networks will be considered in future.
Second,  the extension of resilient average consensus for high-dimension systems with general linear dynamics is left for future work. 
Third, applications of resilient average consensus including formation control and flocking of multi-robot systems can be possible directions.

\bibliographystyle{unsrt}        
\bibliography{autosam}           



\end{document}